\documentclass[aps,prx,twocolumn,superscriptaddress]{revtex4}
\usepackage{graphicx}
\usepackage{amssymb}
\usepackage{amsmath}
\usepackage{epsfig}
\usepackage{color}
\usepackage{amsthm}
\usepackage{float}
\usepackage[colorlinks=true,
			linkcolor=blue,
			urlcolor=cyan,
			citecolor=blue]{hyperref}
\usepackage{hyperref}
\usepackage{lipsum}		% Can be removed after putting your text content
\usepackage{ulem}
\newcommand{\norm}[1]{\left\| #1 \right\|}

\newtheorem{Proposition}{Proposition}

\begin{document}
\title{Training-set-free two-stage deep learning for spectroscopic data de-noising}
\author{Dongchen Huang}
\affiliation{Beijing National Laboratory for Condensed Matter Physics and Institute of
Physics, Chinese Academy of Sciences, Beijing 100190, China}
\affiliation{University of Chinese Academy of Sciences, Beijing 100049, China}
\author{Junde Liu }
\affiliation{Beijing National Laboratory for Condensed Matter Physics and Institute of
Physics, Chinese Academy of Sciences, Beijing 100190, China}
\affiliation{University of Chinese Academy of Sciences, Beijing 100049, China}
\author{Tian Qian}
\email[Corresponding author:]{tqian@iphy.ac.cn}
\affiliation{Beijing National Laboratory for Condensed Matter Physics and Institute of
Physics, Chinese Academy of Sciences, Beijing 100190, China}
\affiliation{University of Chinese Academy of Sciences, Beijing 100049, China}
\affiliation{Songshan Lake Materials Laboratory, Dongguan, Guangdong 523808, China}
\author{Hongming Weng}
\email[Corresponding author:]{hmweng@iphy.ac.cn}
\affiliation{Beijing National Laboratory for Condensed Matter Physics and Institute of
Physics, Chinese Academy of Sciences, Beijing 100190, China}
\affiliation{University of Chinese Academy of Sciences, Beijing 100049, China}
\affiliation{Songshan Lake Materials Laboratory, Dongguan, Guangdong 523808, China}

\date{\today}
\begin{abstract}
    De-noising is a prominent step in the spectra post-processing procedure. Previous machine learning-based methods are fast but mostly based on supervised learning and require a training set that may be typically expensive in real experimental measurements. Unsupervised learning-based algorithms are slow and require many iterations to achieve convergence. Here, we bridge this gap by proposing a training-set-free two-stage deep learning method. We show that the fuzzy fixed input in previous methods can be improved by introducing an adaptive prior. Combined with more advanced optimization techniques, our approach can achieve five times acceleration compared to previous work. Theoretically, we study the landscape of a corresponding non-convex linear problem, and our results indicates that this problem has benign geometry for first-order algorithms to converge.

\end{abstract}

\maketitle
\section{Introduction}\label{section1}
With the continuous development of modern experimental techniques, spectroscopic data has become increasingly complex. The amount of data has also grown significantly, such as in angle-resolved photoemission spectroscopy (ARPES) \cite{Damascelli2003, Lu2012, Hashimoto2014, Yang2018, Zhou2018, Lv2019, Zhang2020, Lv2021, Sobota2021, Zhang2022, Orenstein2012, Aoki2014, Smallwood2016, Bao2022, Torre2021}. To extract fine band structures for discovering new physical features, high-quality data are needed. However, high signal-to-noise-ratio (SNR) measurements demand a long spectrum acquisition time which is even impossible to achieve in some cases. For instance, the duration of ARPES measurements is generally constrained due to factors such as limited synchrotron light resources and the aging of sample surfaces resulting from the adsorption of residual gas molecules. 

Therefore, it is necessary to seek post-processing methods to obtain high-quality data or to highlight spectral characteristics \cite{Zhang2011, He2017, Peng2020}. One example is the removal of noise from spectra. In the past, mathematical methods such as Gaussian smoothing and Fourier transform filtering were commonly employed for noise reduction processing. However, their practical effectiveness was not very satisfactory, as over-processing may lead to the loss of intrinsic spectral information.

Motivated by the rapid development of machine learning techniques in computer vision \cite{He2016}, natural language processing \cite{Brown2020} and condensed matter physics \cite{Zhang2018}, many attempts have been made to remove the measurement noise in spectra \cite{Kim2021, Konstantinova2021, Restrepo2022, Joucken2022, Oppliger2022, Huang2023, Liu2023}. These works can be divided into two lines. One is supervised-learning-based algorithms \cite{Kim2021, Konstantinova2021, Restrepo2022, Joucken2022, Oppliger2022} and these methods require a high-quality training set that is often inaccessible in scientific applications. In addition, it has been reported that the supervised-models may show lack of robustness \cite{Goodfellow2014} and hallucination \cite{Gottschling2020,Bhadra2021} especially if they work outside the training domain or the training data lacks diversity.  

The difficulties of constructing a high-quality training set and the drawbacks of supervised learning techniques have motivated another line of work based on unsupervised learning \cite{Huang2023,Liu2023}. These works utilize the intrinsic self-correlation within a single spectral measurement, and thus they are able to work without access to a training set. Although these methods can achieve promising performance in spectral de-noising comparable to supervised learning counterparts. However, they still have two drawbacks. One is the fuzzy fixed input for all spectra lacking interpretation. The other drawback is the slow speed of convergence. Many iterations are usually necessary to achieve convergence in practice. This slowness largely limits the widespread application of unsupervised learning. 

In this manuscript, we aim to address the above two weaknesses by proposing a two-stage unsupervised de-noising algorithm. To improve the interpretability, we suggest an adaptive prior by employing the linear correlation of the spectra through principal component pursuit (PCP) \cite{Candes2011}. Then the remaining non-linear correlation becomes the learning goal of the neural networks. Combined with more advanced adaptive optimization algorithms that respect the loss landscape, we can reduce the number of iterations by a factor of four times. 

For further justification of our method, we theoretically investigate the landscape of the corresponding non-convex simplified linear model and show that all the critical points are either strict saddle points \cite{Lee2016,Jin2017} or global minima. This benign geometry provides positive conditions for first-order methods to achieve global convergence in minimizing non-convex functions \cite{Lee2016,Lee2017,Panageas2016}.

Finally, we believe that our method can expand the application scope of the sample-efficient unsupervised learning techniques in scientific image processing and our analysis for optimization is not limited to this specified problem and may be of independent interest for other non-convex problems.

\begin{figure*}[ht]	
	\centering
	\includegraphics[width=0.95\textwidth]{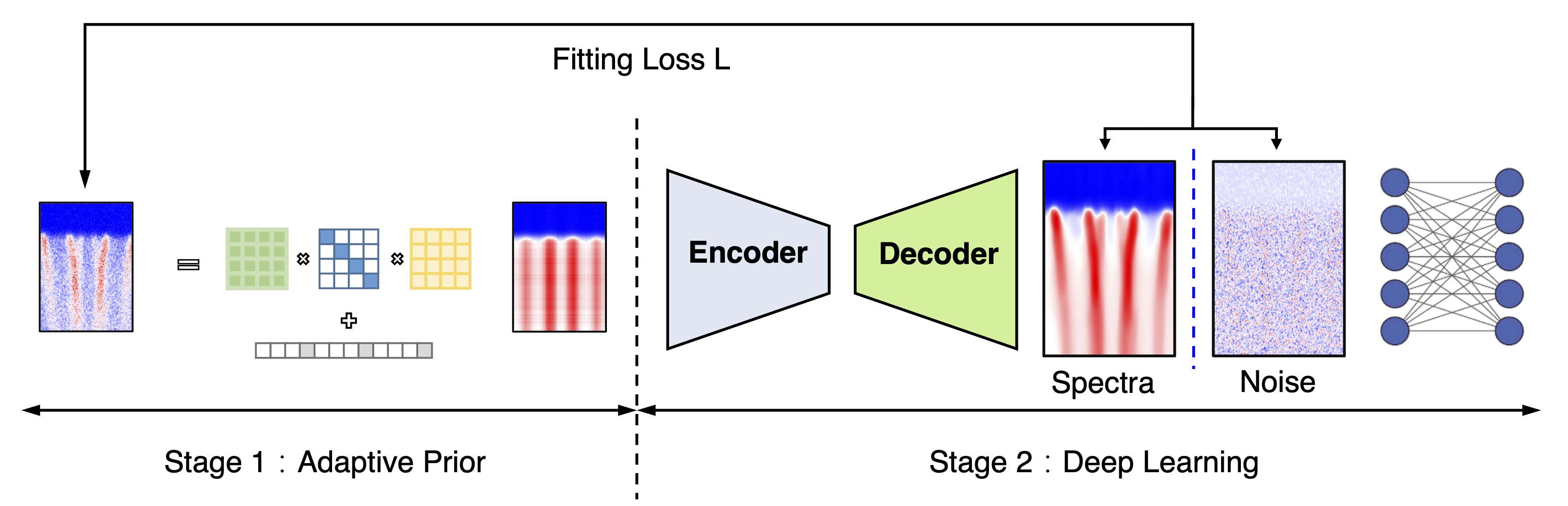}			
	\caption{The illustration of our two-stage deep learning de-noising algorithm. The first stage is building an adaptive prior for individual spectral via principal component pursuit. The adaptive prior only considers the linear correlation. The resulting prior becomes the input of an encoder-decoder network in the second stage. The sparse noise is also parameterized by a small neural network. Both neural networks are trained to predict the clean spectra and noise simultaneously.}
	\label{Fig:illustration of algorithm}
\end{figure*}

\section{Method}\label{sec:2}
A noisy ARPES image contains $x-$ and $y-$ axes representing momentum and energy respectively. Our method is based on the so-called additive noise assumption, and this assumption enables a decomposition of a noisy measurement into the summation of clean spectra and measurement noise. 

Mathematically, let $I \in \mathbb{R}^{m\times n}$ denote the noisy ARPES image where $m$ and $n$ stands for the number of discretized energy and momentum and the value of pixel represents intensity, respectively. Our algorithm aims to find a decomposition which yields:
\begin{equation}
	I = A + S,
	\label{Eq:decomposition}
\end{equation}
where $A$ and $S$ are the desired correlated and sparse parts respectively. This decomposition is obviously overparameterized and not identical. To obtain a reasonable decomposition, we leverage the correlation and sparsity properties of the clean spectra and measurement noise respectively. In other words, the pixels concerning the clean spectra are highly correlated within themselves and the noise is not only distributed in a dotted pattern but also unable to corrupt all the intrinsic energy band information. 

In our manuscript, we solve the decomposition \eqref{Eq:decomposition} by parameterizing both terms with individual neural networks. By leveraging the sparse nature, the noise part $S$ can be parameterized as $g\circ g - h\circ h$, because recent advances in non-convex optimization \cite{You2020} show that this parameterization scheme is easy to solve by first-order methods and encourage sparse solutions. 

The correlation of clean spectra can be well-captured by deep convolutional neural networks (CNN), and we thus parameterize the desired clean spectra with CNN. Practically, previous unsupervised learning practices in ARPES de-nosing \cite{Huang2023} use a predefined input to all spectral images without considering the shape of intrinsic bands. However, this hand-chosen input does not contain any shape about the clean spectra, making the neural network require more iteration to learn. Easing the learning task by constructing a more informative input can enable some acceleration of training.

To reduce the difficulty of learning, as shown in Fig. \ref{Fig:illustration of algorithm}, we directly parameterize the clean spectra via a U-shape encoder-decoder network with convolutional layers and equip this neural network with an \textit{adaptive prior} capturing the basic features of the desired clean spectra. Then, the neural network only needs to learn the correction to the basic outlines rather than the whole details. In addition, we note that the common stochastic gradient descent (SGD) optimizer does not take much curvature information into account. Thus, replacing SGD with an \textit{adaptive optimizer} with momentum can further reduce the number of training steps.

\paragraph*{Adaptive prior.}
A reasonable adaptive prior should satisfy two properties. The first property is interpretability, and the second one is ease of calculation. Fortunately, the two properties are satisfied if we enforce matrix $A$ in Eq. \eqref{Eq:decomposition} representing only linear correlation. Under this assumption, the matrix $A$ becomes low-rank, and we denote it $L$ for emphasizing this property, thus the decomposition \eqref{Eq:decomposition} reduces to the so-called ``low-rank + sparsity" decomposition which has applications in complex networks \cite{Ozdenur2017} and computer vision \cite{Dong2014}. The ``low-rank + sparsity" composition can be obtained by solving the following optimization problem:
\begin{equation}
	\min_{L,S} \norm{L}_* + \lambda \norm{S}_1, \quad s.t. \quad L + S = I,
	\label{Eq:RPCA}
\end{equation}
where $\norm{L}_*$ represents the nuclear norm defined as the summation of all the singular values, $\norm{S}_1$ stands for $\ell^1$ norm defined as the summation of the absolute value of every entry in matrix $S$, and $\lambda$ is a tunable positive parameter corresponding to the amplitude of measurement noise. The linear constrained optimization problem \eqref{Eq:RPCA} is named principal component pursuit and has convexity property as the summation of two convex functions. Thus, the problem \eqref{Eq:RPCA} can be solved globally via the alternating direction method of multipliers (ADMM) \cite{Boyd2011} and all the details are given in the supplemental material. 

\begin{figure}[ht]
	\centering
	\includegraphics[width=0.45\textwidth]{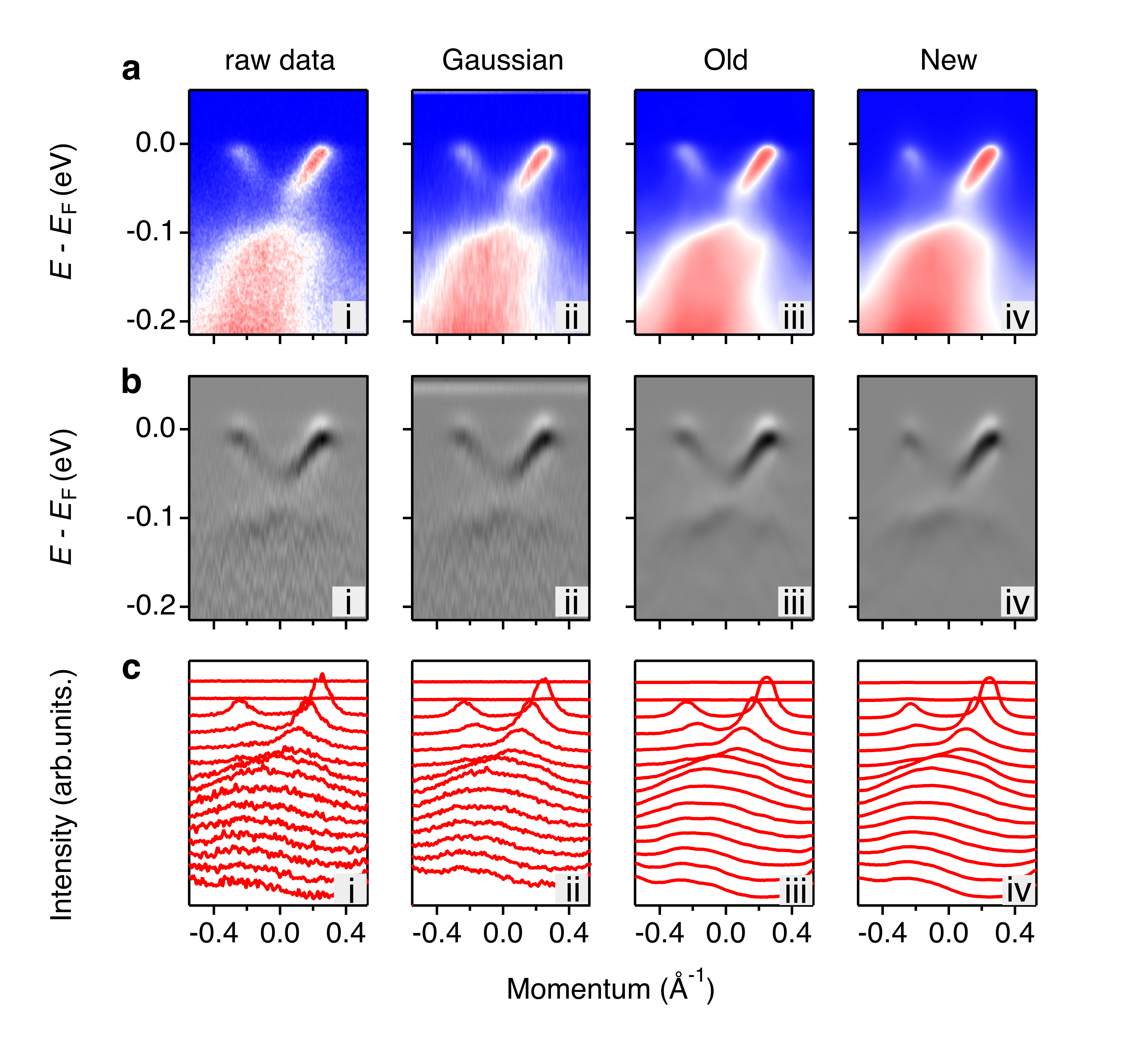}
	\caption{Performance of the de-noised results by different methods. (a) ARPES intensity spectra in FeSe thin film at the M point. (i) The low-quality raw spectra used for the following de-noising process. (ii)-(iv) The de-noised results using (ii) Gaussian smoothing method, (iii) Previous CNNs method, and New CNNs method. (b) Second-derivative plots from the corresponding spectra in panel (a). (c). Momentum distribution curves (MDCs) from the corresponding spectra in panel (a). }
	\label{Fig: De-noised result}
\end{figure}
\paragraph*{Loss function.} 

Combining the above thoughts together, our algorithm aims to remove the noise by minimizing the following loss function:
\begin{equation}
	\begin{split}
		\mathcal{L} &= \norm{A_\theta (L) + g\circ g- h\circ h -I}_2^2,  
		\label{Eq:Loss}
	\end{split}
\end{equation}
where $\theta$ denotes the collection of neural network parameters to be optimized and $L$ is both the input of the neural network and the solution of PCP decomposition (Eq. \eqref{Eq:RPCA}).

\paragraph*{Adaptive Optimization.}
Optimization respecting the landscape can no doubt accelerate the training and reduce the number of iterations \cite{Martens2020}. Stochastic gradient descent can be accelerated by introducing a conditioner in front of the gradient term. However, the best conditioner (Fisher matrix) is notoriously hard for precise estimation \cite{Martens2020}. Fortunately, a diagonal approximation to the Fisher matrix can be implemented by adaptive momentum methods \cite{Martens2020} and is able to achieve reasonable performance in deep learning practice. Thus, we introduce AdamW \cite{Loshchilov2018} optimizer to update all parameters.

Thus, we need to minimize the loss function \eqref{Eq:Loss} by the following update rule:
\begin{equation}
	\begin{split}
		\theta_{t+1} &= \theta_t - \eta_a \left( \frac{\hat{v}_{\theta t}}{\sqrt{\hat{s}_{\theta t}} + \epsilon} + \lambda_1 \theta_t \right) \\
		g_{t+1} &= g_t - \eta_s \left( \frac{\hat{v}_{g t}}{\sqrt{\hat{s}_{g t}} + \epsilon} + \lambda_2 g_t \right)  \\
		h_{t+1} &= h_t - \eta_s \left( \frac{\hat{v}_{h t}}{\sqrt{\hat{s}_{h t}} + \epsilon} + \lambda_2 h_t \right),
	\end{split}
\end{equation}
where $\eta_s$ is the learning rate, $\hat{v}_{\theta t},\hat{v}_{g t},\hat{v}_{h t},\hat{s}_{\theta t},\hat{s}_{g t}$ and $\hat{s}_{h t}$ are momentum parameters of the optimizer depending on the history of gradient and its square, and detail calculation is given in the appendix. We use weight decay with coefficient $\lambda_1,\lambda_2>0$ to reduce the damping throughout the whole training procedure. The most crucial role here is the discrepant learning rate for the correlation and noise part. We will show in Sec. \ref{sec:4} that this is the key tunable parameter in our algorithm.

\section{Applications}\label{sec:3}

\subsection{The performance of de-noising}

As an illustrative example, we have applied our novel approach to the ARPES spectra of a mono-layer of FeSe at the M point. The dispersion spectra of the raw data and those after employing various de-noising methods are compared in Fig. \ref{Fig: De-noised result}(a). In this figure, we compare our newly developed technique with both Gaussian smoothing and the method previously described in \cite{Huang2023}. The original data exhibits a low signal-to-noise (SNR) due to the abbreviated measurement duration. Employing moderate Gaussian smoothing only mitigates a portion of the high-frequency noise, leaving the overall spectra still relatively uneven. Conversely, increasing the degree of smoothing excessively blurs the spectra, obscures the band features, and sacrifices some intricate details. Consequently, Gaussian smoothing proves ineffective in dealing with sparse noise. In contrast, both our innovative approach and the aforementioned deep learning methods demonstrate commendable performance. Our new method achieves comparable results as the previous deep learning algorithms, where the noise can be effectively eliminated from the spectra while retaining the characteristic energy band dispersion. 

The second derivative of the spectra is depicted in Fig. \ref{Fig: De-noised result}(b) to enhance the clarity of the band dispersions, which is a commonly employed technique in ARPES data analysis \cite{Zhang2011}. Once Again, the quality of the band structures remains less than optimal and lacks smoothness in the raw spectra, as evident in Fig. \ref{Fig: De-noised result}(b-i).  Furthermore, the Gaussian smoothing method fails to impart any significant enhancement. By contrast, the second derivative of our novel de-noised data exhibits remarkably distinct and smooth band dispersions. The sensitivity of the second derivative to noise underscores the substantial impact of noise removal on performance. Consequently, more effective de-noising greatly facilitates the clear visualization of the energy bands. This underscores the efficacy of our deep-learning-based de-noising methods for post-processing, as they distinctly contribute to achieving clearer spectra. 

Furthermore, in Fig. \ref{Fig: De-noised result}(c), we present momentum distribution curves (MDCs) of the spectra to provide a more intuitive demonstration of the effectiveness of our novel de-noising approach. The original data, as depicted in Fig. \ref{Fig: De-noised result}(c-i), exhibit pronounced noise. Although Gaussian smoothing manages to eliminate high-frequency noise, it leaves behind low-frequency noise in the MDCs. This is evident in Fig. \ref{Fig: De-noised result}(c-ii), where the curves still exhibit a noisy character. More critically, with an increase in the degree of smoothing, the peak width also expands, consequently erasing the intrinsic information regarding the energy bands and posing challenges for subsequent quantitative analyses, such as self-energy extraction. By contrast, as illustrated in Fig. \ref{Fig: De-noised result}(c-iii) and Fig. \ref{Fig: De-noised result}(c-iv), the de-noised results obtained through both our novel approach and the previously employed deep-learning method present remarkable smoothness. These results preserve the intrinsic features of the band structures, including the peak position and width, thereby safeguarding crucial features for comprehensive analysis.

\begin{figure}[ht]
	\centering
	\includegraphics[width=0.45\textwidth]{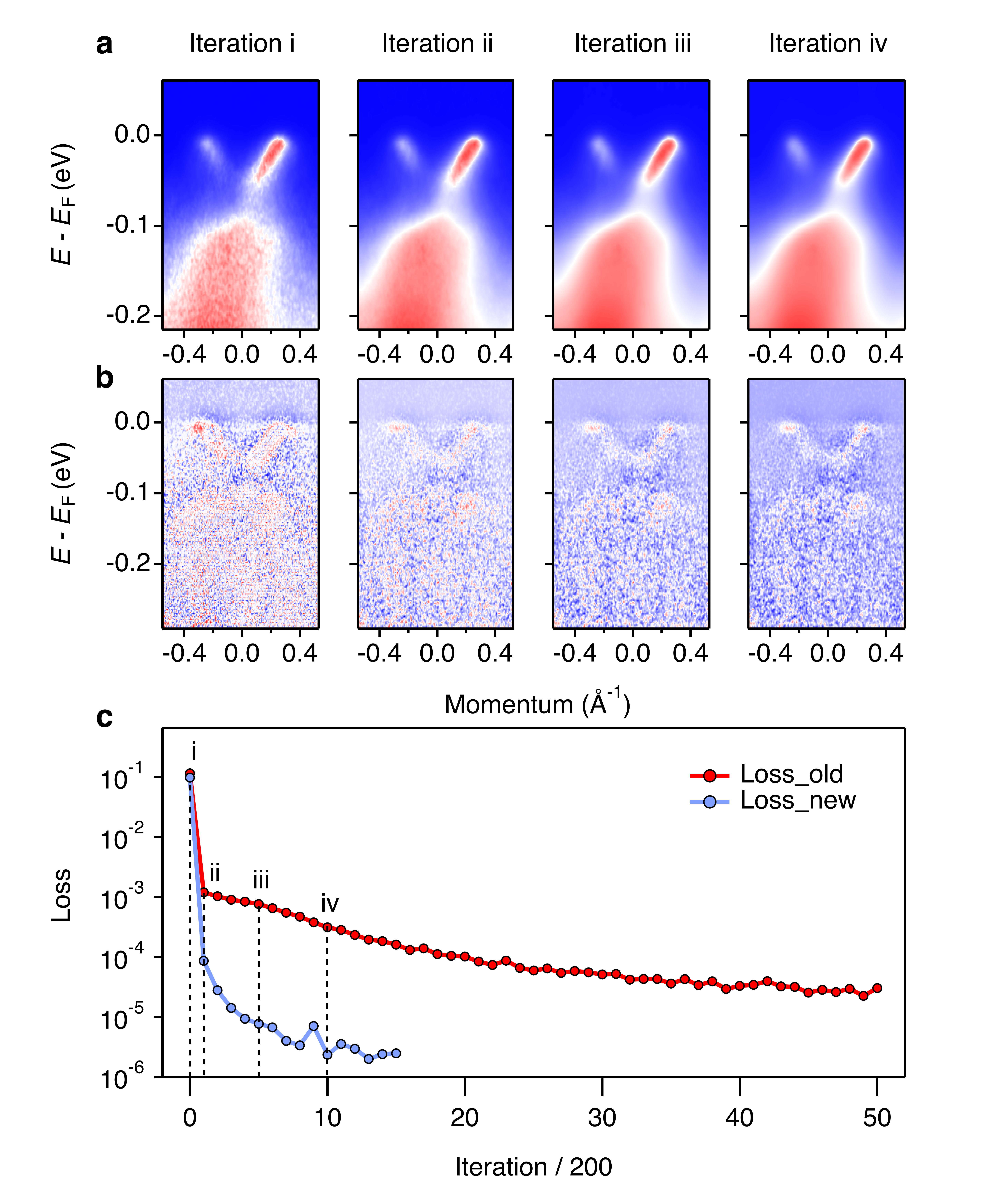}
	\caption{De-noised results of FeSe thin film at the M point under different iteration steps and the training curve at different stages. (a) De-noised results under four checkpoints corresponding to the iteration step 0, 200, 1000, and 2000. (b) The corresponding noise of panel (a). (c) Validation loss as a function of the iteration number. The blue curve and red curve indicate our new method and previous method, respectively. After around 2000 interactions, the loss of our new method converges which is faster than the previous method.}
	\label{Fig: De-noised speed}
\end{figure}

\subsection{Training dynamics and processing speed}

The de-noised results of FeSe thin film at the M point and training progression at various stages are illustrated in Fig. \ref{Fig: De-noised speed}. It can be seen that as the iteration of the computing process increases, the intrinsic energy band signals are gradually extracted [Fig. \ref{Fig: De-noised speed}(a)], leaving only the noise structure [Fig. \ref{Fig: De-noised speed}(a)]. Throughout the denoising training process, the neural network initially acquires low-frequency, long-range features before subsequently incorporating high-frequency, localized details. Thanks to the algorithm's prior processing, the new method can swiftly extract the characteristics of the energy bands, resulting in a faster completion of the denoising process compared to the previous method. From the loss function dynamics shown in Fig. \ref{Fig: De-noised speed}(c), it is more evident that previous methods demanded 8000 iterations for convergence, whereas our new approach accomplishes the same in just 2000 iterations. This highlights the remarkable efficiency of our new algorithm in removing the noise of data, leading to a substantial enhancement in experimental efficiency. We note by passing that conducting the input prior do not influence the convergence speed and using adaptive momentum optimizer require tedious cause hyperparameter tuning, but combining the two tricks together work reasonably well.

\section{Discussion}\label{sec:4}
The performance of our algorithm may be attributed to the two following reasons: fine-tuned leaning rate ratio $\eta_a/\eta_s$ and benign geometry.

\subsection{The effect of discrepant learning rate}
The learning rate ratio $\eta_a/\eta_s$ plays a key role in our algorithm. An illustration of the effect of the learning ratio is shown in Fig. \ref{Fig:Learning rate}, where we fix $\eta_a = 0.02$ and change the noise learning rate $\eta_s$. Only a proper ratio $0.5$ can recover an appropriate spectral image. An excessively large ratio encourages our algorithm to view some intrinsic band structures as noise, leading to an over-smooth spectral. On the contrary, a tiny learning rate ratio may force the algorithm to retain too much unwanted information and prevent the algorithm from identifying and removing all the measurement noise.

\begin{figure}[ht]
	\centering
	\includegraphics[width=0.45\textwidth]{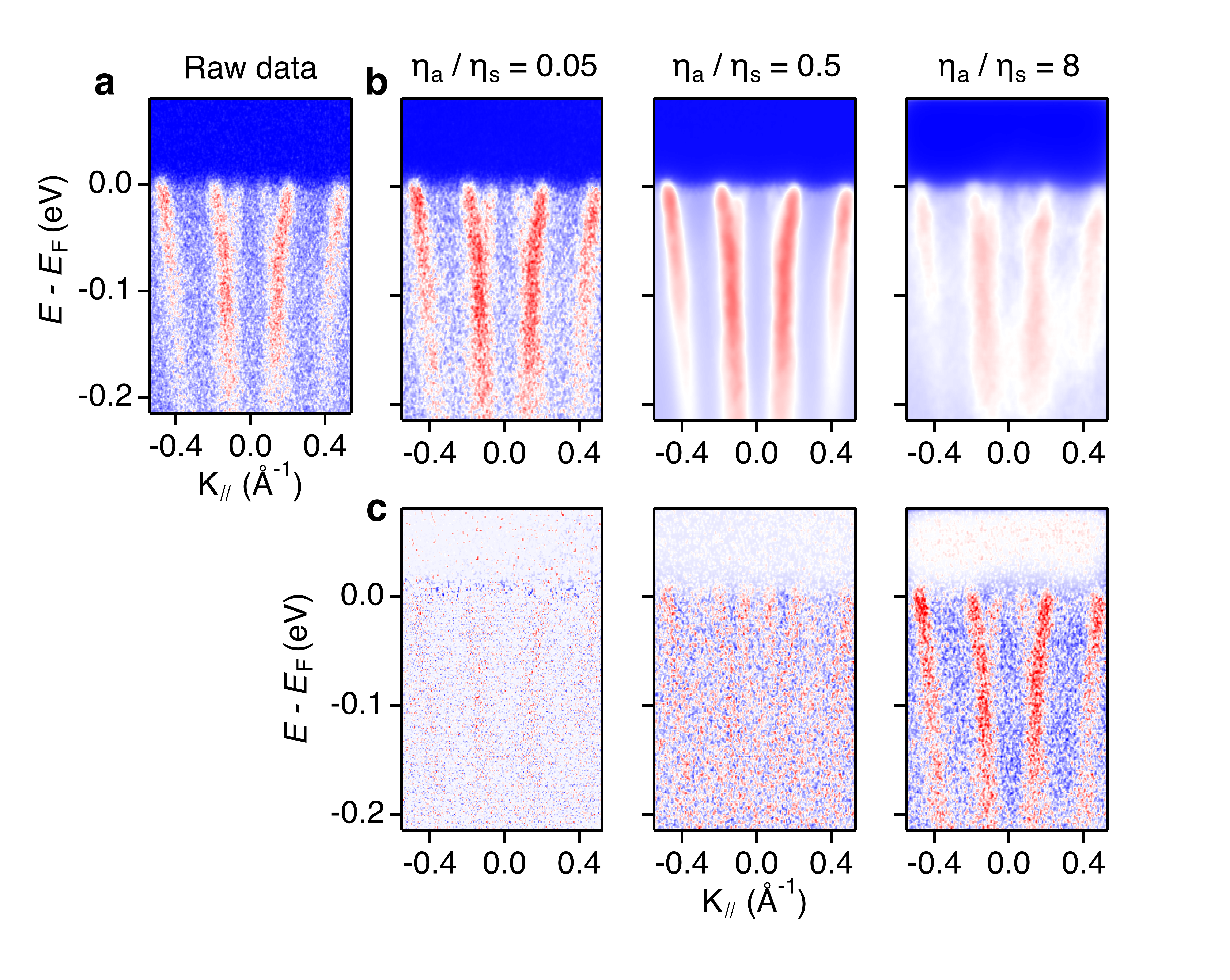}
	\caption{De-noised results under different noise learning rate ratios. (a) The raw data of Bi2212 along the nodal cut. (b) The de-noised data based on the original data in panel (a) using different noise learning rate ratios $\eta_a/\eta_s = 0.05,0.5,8$. (c) The corresponding noise of panel (b).}
	\label{Fig:Learning rate}
\end{figure}

The above phenomena about discrepant learning rates may be interpreted by the implicit bias of gradient descent \cite{You2020}. Precise analysis of practical neural networks is notoriously hard due to their non-linearity and non-convexity, thus we begin with the following non-convex simplified model:
\begin{equation}
	\min_{U,g,h} \mathcal{L}'(U,g,h) = \norm{UU^T + g\circ g - h \circ h -I}_2^2,
	\label{Eq:Simplified model}
\end{equation}
where $U$ is a real matrix and $U^T$ denotes its transpose. This model linearizes the non-linear encoder and the decoder network via the matrix $U$ and $U^T$ respectively while keeping both the non-convexity of the neural network and the special parameterization scheme of sparse noise.

It has been proved \cite{You2020} that the solution of the optimization problem \eqref{Eq:Simplified model} is identical to the principal component pursuit problem \eqref{Eq:RPCA} under mild conditions with $L = UU^T,s = g\circ g - h \circ h$ and $\lambda = \eta_a/\eta_s$. Compared this property with the results in Fig. \ref{Fig:Learning rate}, we can see that the effect of learning rate does not change in non-linear cases. The predicted noise is somewhat sparse whose amplitude and shape are controlled by the value of $\lambda (\eta_s)$. Extraordinary large learning rates correspond to small $\lambda$ and then the algorithm overestimates the amplitude of noise. Thus, much energy band information is incorrectly viewed as noise, resulting in an over-smooth band structure. Similarly, unreasonably small learning rate settings underestimate the noise amplitude. Then, the algorithm is only able to extract a small fraction of measurement noise, leaving the resulting energy band rough. A satisfactory spectra can only be obtained by setting an appropriate choice of learning rate ratio.

\subsection{Benign Landscape}
During training, good convergence is usually observed in practice, and this may reduce the presence of some favorable landscape of the optimization. The non-convex linearized problem \eqref{Eq:Simplified model} also provides insights into a benign geometry conducive to achieving convergence. We show that the stationary points in Eq. \eqref{Eq:Simplified model} are either global minimal or strict saddle point, and this landscape is beneficial for first-order methods like gradient descent to achieve global convergence.

We first characterize the stationary points of Eq. \eqref{Eq:Simplified model} by obtain taking derivatives of $\mathcal{L}’$ with respect to $U,g$ and $h$, we have:
\begin{equation}
	\begin{split}
		\nabla_U \mathcal{L}' &= (U-M)U + (U-M^T) U = 0 \\
		\nabla_g \mathcal{L}' &= r \circ g = 0 \\
		\nabla_h \mathcal{L}' &= -r \circ h = 0,
	\end{split}
\end{equation}
where $ -M = g\circ g - h\circ h -I$, $r = UU^T + g\circ g - h \circ h -I$ and $\circ$ denotes the element-wise Hadamard product. We must have $r \neq 0$ for non-global minimums, and this implies $g = h =0$ at local minimum or saddle points. This result indicates that the critical points that do not achieve the global minimum only correspond to clean spectra and zero measurement noise, which is not a satisfactory solution. Fortunately, the saddle point solution does not have much influence and all the saddle points are easy to escape due to the following strict saddle point property.

Given a function $\mathcal{L}'$, a saddle point $x$ is \textit{strict} if there exists a direction $d$ satisfying $d^T H(x)d <0$ where $H$ is the Hessian matrix with respect to $\mathcal{L}'$. In other words, the strict saddles have a negative curvature direction, and any perturbation along this direction enables the escape during optimization \cite{Jin2017}.

By carefully investigating the Hessian, we show that the following holds:
\begin{Proposition}
    Any critical point of the linear model \eqref{Eq:Simplified model} is either a global minimum or a strict saddle point whose Hessian has at least one negative eigenvalue.
\end{Proposition}
This proposition has an interesting corollary: \textit{all the local minimums are global} because the Hessian with respect to the local minimums does not have negative eigenvalues and becomes negative-defined. This property has also been found in non-convex optimizations like tensor decomposition \cite{Ge2015} and some neural networks \cite{Choromanska2015}.

Having this benign landscape, the first-order method is able to efficiently escape the saddle point \cite{Jin2017}. One may guess a global convergence for problem \eqref{Eq:Simplified model}. The global convergence can indeed be achieved under random initialization if the gradient function satisfies the global Lipschitz property \cite{Lee2016}. Unfortunately, we cannot directly apply this result here due to the cubic term in the gradient. This unsatisfaction does not mean the first-order achieve cannot achieve global convergence, because the problem \eqref{Eq:Simplified model} can be viewed as a reparameterization of principal component pursuit with $L =UU^T$ and $s = g\circ g - h \circ h$ and the effect of the convexity behind the original problem needs more carefully investigation. For instance, the global Lipschitz property can be relaxed by introducing a special convex domain \cite{Panageas2016,Lee2017}, and the detailed construction of the convex domain remains open.
We believe the interplay between optimization and ``hidden convexity" can not be neglected and leave this as a future work.
 
\section{Conclusion}\label{sec:5}
In this manuscript, we propose a two-stage training-set free deep learning method for spectral data de-noising. Our method can achieve comparable performance and faster convergence than the previous method. Different from other black-box methods, our algorithm explicitly considers the structure of the spectral image in the procedure of input construction. Further, our landscape analysis illustrates that any saddle point of this linear model is either a global minimum or a strict saddle point. Like all the deep-learning methods, our algorithm is easy to extend, this method can be applied to scientific image-processing tasks with correlation and sparse structures.

%%%%%%%%%%%%%%%%%%%%%%%%%%%%%%%%%%%%%%%%%%%%%%%%%%%%%%%
\section*{ACKNOWLEDGEMENTS}
We thank Kaiyuan Cui for useful and insightful discussions. This work was supported by the National Natural Science Foundation of China (Grants No. 11925408, No. 11921004, No. 12188101, and U1832202), the Ministry of Science and Technology of China (Grant No.  2022YFA1403800, and No. 2022YFA1403800), the Chinese Academy of Sciences (Grant No. XDB33000000) and the Informatization Plan of the Chinese Academy of Sciences (Grant No. CASWX2021SF-0102).

D.H. and J.L. contributed equally to this work.

%%%%%%%%%%%%%%%%%%%%%%%%%%%%%%%%%%%%%%%%%%%%%%%%%%%%%%%
%%% Appendix sections. 

%%%%%%%%%%%%%%%%%%%%%%%%%%%%%%%%%%%%%%%%%%%%%%%%%%%%%%
%% Reference section


\begin{thebibliography}{10}

	\bibitem{Damascelli2003}
	A.~Damascelli, Z.~Hussain, and Z.-X. Shen, {Angle-resolved photoemission
	  studies of the cuprate superconductors}, { Rev. Mod. Phys.} \textbf{75},
	  473 (2003).
	
	\bibitem{Lu2012}
	D.~Lu, I.~M. Vishik, M.~Yi, Y.~Chen, R.~G. Moore, and Z.-X. Shen,
	  {Angle-resolved photoemission studies of quantum materials}, { Annu.
	  Rev. Mater. Res.} \textbf{3}, 129 (2012).
	
	\bibitem{Hashimoto2014}
	M.~Hashimoto, I.~M. Vishik, R.-H. He, T.~P. Devereaux, and Z.-X. Shen,
	  {Energy gaps in high-transition-temperature cuprate superconductors},
	  { Nat. Phys.} \textbf{10}, 483 (2014).
	
	\bibitem{Yang2018}
	H.~Yang, A.~Liang, C.~Chen, C.~Zhang, N.~B.~M. Schroeter, and Y.~Chen,
	  {Visualizing electronic structures of quantum materials by angle-resolved
	  photoemission spectroscopy}, { Nat. Rev. Phys.} \textbf{3}, 341 (2018).
	
	\bibitem{Zhou2018}
	X.~Zhou, S.~He, G.~Liu, L.~Zhao, L.~Yu, and W.~Zhang, {New developments in
	  laser-based photoemission spectroscopy and its scientific applications: a key
	  issues review}, { Rep. Prog. Phys.} \textbf{81}, 062101 (2018).
	
	\bibitem{Lv2019}
	B.~Lv, T.~Qian, and H.~Ding, {Angle-resolved photoemission spectroscopy and
	  its application to topological materials}, { Nat. Rev. Phys.} \textbf{1},
	  609 (2019).
	
	\bibitem{Zhang2020}
	C.~Zhang, Y.~Li, D.~Pei, Z.~Liu, and Y.~Chen, {Angle-resolved photoemission
	  spectroscopy study of topological quantum materials}, { Annu. Rev.
	  Mater. Res.} \textbf{50}, 131 (2020).
	
	\bibitem{Lv2021}
	B.~Q. Lv, T.~Qian, and H.~Ding, {Experimental perspective on
	  three-dimensional topological semimetals}, { Rev. Mod. Phys.} \textbf{93},
	  025002 (2021).
	
	\bibitem{Sobota2021}
	J.~A. Sobota, Y.~He, and Z.-X. Shen, {Angle-resolved photoemission studies of
	  quantum materials}, { Rev. Mod. Phys.} \textbf{93}, 025006 (2021).
	
	\bibitem{Zhang2022}
	H.~Zhang, T.~Pincelli, C.~Jozwiak, T.~Kondo, R.~Ernstorfer, T.~Sato, and
	  S.~Zhou, {Angle-resolved photoemission spectroscopy}, { Nat. Rev.
	  Methods. Primers.} \textbf{2}, 54 (2022).
	
	\bibitem{Orenstein2012}
	J.~Orenstein, {Ultrafast spectroscopy of quantum materials}, { Phys.
	  Today}, \textbf{65}, 44 (2012).
	
	\bibitem{Aoki2014}
	H.~Aoki, N.~Tsuji, M.~Eckstein, M.~Kollar, T.~Oka, and P.~Werner,
	  {Nonequilibrium dynamical mean-field theory and its applications}, {
	  Rev. Mod. Phys.} \textbf{86}, 779 (2014).
	
	\bibitem{Smallwood2016}
	C.~L. Smallwood, R.~A. Kaindl, and A.~Lanzara, {Ultrafast angle-resolved
	  photoemission spectroscopy of quantum materials}, { EPL}, \textbf{115},
	  27001 (2016).
	
	\bibitem{Bao2022}
	C.~Bao, P.~Tang, D.~Sun, and S.~Zhou, {Light-induced emergent phenomena in 2D
	  materials and topological materials}, { Nat. Rev. Phys.} \textbf{4}, 33
	  (2022).
	
	\bibitem{Torre2021}
	A.~d.~l. Torre, D.~M. Kennes, M.~Claassen, S.~Gerber, J.~W. McIver, and M.~A.
	  Sentef, {Colloquium: Nonthermal pathways to ultrafast control in quantum
	  materials}, { Rev. Mod. Phys.} \textbf{93}, 041002 (2021).
	
	\bibitem{Zhang2011}
	P.~Zhang, P.~Richard, T.~Qian, Y.-M. Xu, X.~Dai, and H.~Ding, {A precise
	  method for visualizing dispersive features in image plots}, { Rev. Sci.
	  Instrum.} \textbf{82}, 043712 (2011).
	
	\bibitem{He2017}
	Y.~He, Y.~Wang, and Z.-X. Shen, {Visualizing dispersive features in 2D image
	  via minimum gradient method}, { Rev. Sci. Instrum.} \textbf{88}, 073903
	  (2017).
	
	\bibitem{Peng2020}
	H.~Peng, X.~Gao, Y.~He, Y.~Li, Y.~Ji, C.~Liu, S.~A. Ekahana, D.~Pei, Z.~Liu, Z.~Shen, and Y.~Chen, {Super resolution convolutional neural network for
	  feature extraction in spectroscopic data}, { Rev. Sci. Instrum.}
	  \textbf{91}, 033905 (2020).
	
	\bibitem{He2016}
	K.~He, X.~Zhang, S.~Ren, and J.~Sun, Deep residual learning for image
	  recognition, in { \textit{2016 IEEE Conference on Computer Vision and Pattern
	  Recognition (CVPR)}} (IEEE, New York, 2016).
	
	\bibitem{Brown2020}
	T.~Brown, B.~Mann, N.~Ryder, M.~Subbiah, J.~D. Kaplan, P.~Dhariwal,
	  A.~Neelakantan, P.~Shyam, G.~Sastry, A.~Askell, S.~Agarwal, A.~Herbert-Voss,
	  G.~Krueger, T.~Henighan, R.~Child, A.~Ramesh, D.~Ziegler, J.~Wu, C.~Winter,
	  C.~Hesse, M.~Chen, E.~Sigler, M.~Litwin, S.~Gray, B.~Chess, J.~Clark,
	  C.~Berner, S.~McCandlish, A.~Radford, I.~Sutskever, and D.~Amodei, Language
	  models are few-shot learners, in { \textit{Advances in Neural Information
	  Processing Systems}} (Curran Associates, Inc., Red Hook, NY, 2020).
	
	\bibitem{Zhang2018}
	L.~Zhang, J.~Han, H.~Wang, R.~Car, and W.~E, Deep potential molecular
	  dynamics: A scalable model with the accuracy of quantum mechanics, {
	  Phys. Rev. Lett.} \textbf{120}, 143001 (2018).
	
	\bibitem{Kim2021}
	Y.~Kim, D.~Oh, S.~Huh, D.~Song, S.~Jeong, J.~Kwon, M.~Kim, D.~Kim, H.~Ryu,
	  J.~Jung, W.~Kyung, B.~Sohn, S.~Lee, J.~Hyun, Y.~Lee, Y.~Kim, and C.~Kim,
	  {Deep learning-based statistical noise reduction for multidimensional
	  spectral data}, { Rev. Sci. Instrum.} \textbf{92}, 073901 (2021).
	
	\bibitem{Konstantinova2021}
	T.~Konstantinova, L.~Wiegart, M.~Rakitin, A.~M. DeGennaro, and A.~M. Barbour,
	  {Noise reduction in X-ray photon correlation spectroscopy with
	  convolutional neural networks encoder–decoder models}, { Sci. Rep.}
	  \textbf{11}, 14756 (2021).
	
	\bibitem{Restrepo2022}
	F.~Restrepo, J.~Zhao, and U.~Chatterjee, {Denoising and feature extraction in
	  photoemission spectra with variational auto-encoder neural networks}, {
	  Rev. Sci. Instrum.} \textbf{93}, 065106 (2022).
	
	\bibitem{Joucken2022}
	F.~Joucken, J.~L. Davenport, Z.~Ge, E.~A. Quezada-Lopez, T.~Taniguchi,
	  K.~Watanabe, J.~Velasco, J.~Lagoute, and R.~A. Kaindl, {Denoising scanning
	  tunneling microscopy images with machine learning}, { arxiv:2206.08951}.
	
	\bibitem{Oppliger2022}
	J.~Oppliger, M.~M. Denner, J.~Küspert, R.~Frison, Q.~Wang, A.~Morawietz,
	  O.~Ivashko, A.-C. Dippel, M.~v. Zimmermann, N.~B. Christensen, T.~Kurosawa,
	  N.~Momono, M.~Oda, F.~D. Natterer, M.~H. Fischer, T.~Neupert, and J.~Chang,
	  {Weak-signal extraction enabled by deep-neural-network denoising of
	  diffraction data}, { arXiv:2209.09247}.
	
	\bibitem{Huang2023}
	D.~Huang, J.~Liu, T.~Qian, and Y.-F. Yang, {Spectroscopic data de-noising via
	  training-set-free deep learning method}, { Sci. China. Phys. Mech.
	  Astronomy.} \textbf{66}, 267011 (2023).
	
	\bibitem{Liu2023}
	J.~Liu, D.~Huang, Y.-f. Yang, and T.~Qian, {Removing grid structure in
	  angle-resolved photoemission spectra via deep learning method}, { Phys.
	  Rev. B} \textbf{107}, 165106 (2023).
	
	\bibitem{Goodfellow2014}
	I.~J. Goodfellow, J.~Shlens, and C.~Szegedy, Explaining and harnessing
	  adversarial examples, in { \textit{International Conference on Learning
	  Representations}}, edited by Y.~Bengio and Y.~LeCun (San Diego, 2015).
	
	\bibitem{Gottschling2020}
	N.~M. Gottschling, V.~Antun, B.~Adcock, and A.~C. Hansen, The troublesome
	  kernel: why deep learning for inverse problems is typically unstable, {
	  arXiv:2001.01258}.
	
	\bibitem{Bhadra2021}
	S.~Bhadra, V.~A. Kelkar, F.~J. Brooks, and M.~A. Anastasio, On hallucinations
	  in tomographic image reconstruction, { IEEE Trans. Med. Imag.} \textbf{40},
	  3249 (2021).
   
	\bibitem{Candes2011} E. J. Cand$\grave{e}$s, X. Li, Y. Ma, and J. Wright, Robust principal component analysis?  J. Assoc. Comput. Mach. {\bf 58}, 1 (2011).
 
	\bibitem{Lee2016}
	J.~D. Lee, M.~Simchowitz, M.~I. Jordan, and B.~Recht, Gradient descent only
	  converges to minimizers, in { \textit{Proceedings of the 29th Annual Conference on Learning Theory}}, edited by  V. Feldman, A. Rakhlin, and O. Shamir (PMLR, New York, New York, USA, 2016).
	
	\bibitem{Jin2017}
	C.~Jin, R.~Ge, P.~Netrapalli, S.~M. Kakade, and M.~I. Jordan, How to escape
	  saddle points efficiently, in {\it Proceedings of the 34th International
	  Conference on Machine Learning}, edited by  D. Precup and Y. W. Teh (PMLR, Sydney, 2017).
	
	\bibitem{Lee2017}
	J.~D. Lee, I.~Panageas, G.~Piliouras, M.~Simchowitz, M.~I. Jordan, and
	  B.~Recht, First-order methods almost always avoid strict saddle points,
	  { Math. Program.} \textbf{176}, 11 (2019).
	
	\bibitem{Panageas2016}
	I.~Panageas and G.~Piliouras, Gradient descent only converges to minimizers:
	  Non-isolated critical points and invariant regions, in {\textit{8th Innovations in Theoretical Computer Science Conference (ITCS 2017)}}, edited by C. H. Papadimitriou,  (Leibniz Zentrum für Informatik, Schloss-Dagstuhl, 2017).
	
	\bibitem{You2020}
	C.~You, Z.~Zhu, Q.~Qu, and Y.~Ma, Robust recovery via implicit bias of
	  discrepant learning rates for double over-parameterization, in {
	  \textit{Advances in Neural Information Processing Systems}}, (Curran Associates, Inc.,  Red Hook, NY, 2020).
	
	\bibitem{Ozdenur2017}
	A.~Ozdemir, E.~M. Bernat, and S.~Aviyente, Recursive tensor subspace tracking
	  for dynamic brain network analysis, { IEEE Trans. Signal Inf. Process.
	  Over Netw.} \textbf{3}, 669 (2017).
	
	\bibitem{Dong2014}
	W.~Dong, G.~Shi, X.~Hu, and Y.~Ma, Nonlocal sparse and low-rank
	  regularization for optical flow estimation, { IEEE Trans. Image
	  Process.} \textbf{23}, 4527 (2014).
	
	\bibitem{Boyd2011}
	S.~Boyd, N.~Parikh, E.~Chu, B.~Peleato, and J.~Eckstein, Distributed optimization and statistical learning via the alternating direction
method of multipliers, 	\newblock Foundations and Trends in Machine Learning \textbf{3}, 1 (2010).
	
	\bibitem{Martens2020}
	J.~Martens, New insights and perspectives on the natural gradient method,
	  { J. Mach. Learn. Res.} \textbf{21}, 1 (2020).
	
	\bibitem{Loshchilov2018}
	I.~Loshchilov and F.~Hutter, Decoupled weight decay regularization, in {
	  \textit{International Conference on Learning Representations}}, (New Orleans, 2019).
	
	\bibitem{Ge2015}
	R.~Ge, F.~Huang, C.~Jin, and Y.~Yuan, Escaping from saddle points --- online
	  stochastic gradient for tensor decomposition, in {\it Proceedings of The
	  28th Conference on Learning Theory}, edited by P. Grünwald, E. Hazan, and  S. Kale (PMLR, Paris, France, 2015).
	
	\bibitem{Choromanska2015}
	A.~Choromanska, M.~Henaff, M.~Mathieu, G.~Ben~Arous, and Y.~LeCun, {The loss surfaces of multilayer networks}, in {\it Proceedings of the Eighteenth International Conference on Artificial Intelligence and Statistics}, (PMLR, San Diego, California, USA, 2015)
	
	\end{thebibliography}
\end{document}

\title{Training-set-free two-stage deep learning for Spectroscopic data de-noising\\
\vspace{0.2cm}
- Supplemental Material -}
\author{Dongchen Huang}
\affiliation{Beijing National Laboratory for Condensed Matter Physics and Institute
of Physics, Chinese Academy of Sciences, Beijing 100190, China}
\affiliation{University of Chinese Academy of Sciences, Beijing 100049, China}
\author{Junde Liu}
\affiliation{Beijing National Laboratory for Condensed Matter Physics and Institute
of Physics, Chinese Academy of Sciences, Beijing 100190, China}
\affiliation{University of Chinese Academy of Sciences, Beijing 100049, China}
\author{Tian Qian}
\email[]{tqian@iphy.ac.cn}
\affiliation{Beijing National Laboratory for Condensed Matter Physics and Institute
of Physics, Chinese Academy of Sciences, Beijing 100190, China}
\affiliation{School of Physical Sciences, University of Chinese Academy of Sciences, Beijing 100049, China}
\affiliation{Songshan Lake Materials Laboratory, Dongguan, Guangdong 523808, China}
\author{Hongming Weng}
\email[]{hmweng@iphy.ac.cn}
\affiliation{Beijing National Laboratory for Condensed Matter Physics and Institute
of Physics, Chinese Academy of Sciences, Beijing 100190, China}
\affiliation{School of Physical Sciences, University of Chinese Academy of Sciences, Beijing 100049, China}
\affiliation{Songshan Lake Materials Laboratory, Dongguan, Guangdong 523808, China}
\maketitle

\subsection{I. More results about noise removing}\label{sec:S1}

The Fig. \ref{Fig: More results} displays additional examples of ARPES spectra in which the noise has been successfully eliminated based on the correlation information within the data. It is evident that all the de-noised spectra exhibit high signal-to-noise (SNR) levels, resulting in clearer energy band structures. This result shows the effectiveness of our method in extracting noise for different kinds of spectra solely through the self-correlation information of a single spectrum. Consequently, it further confirms the universality and versatility of our approach.

\begin{figure*}[ht]	
	\centering
	\includegraphics[width=0.95\textwidth]{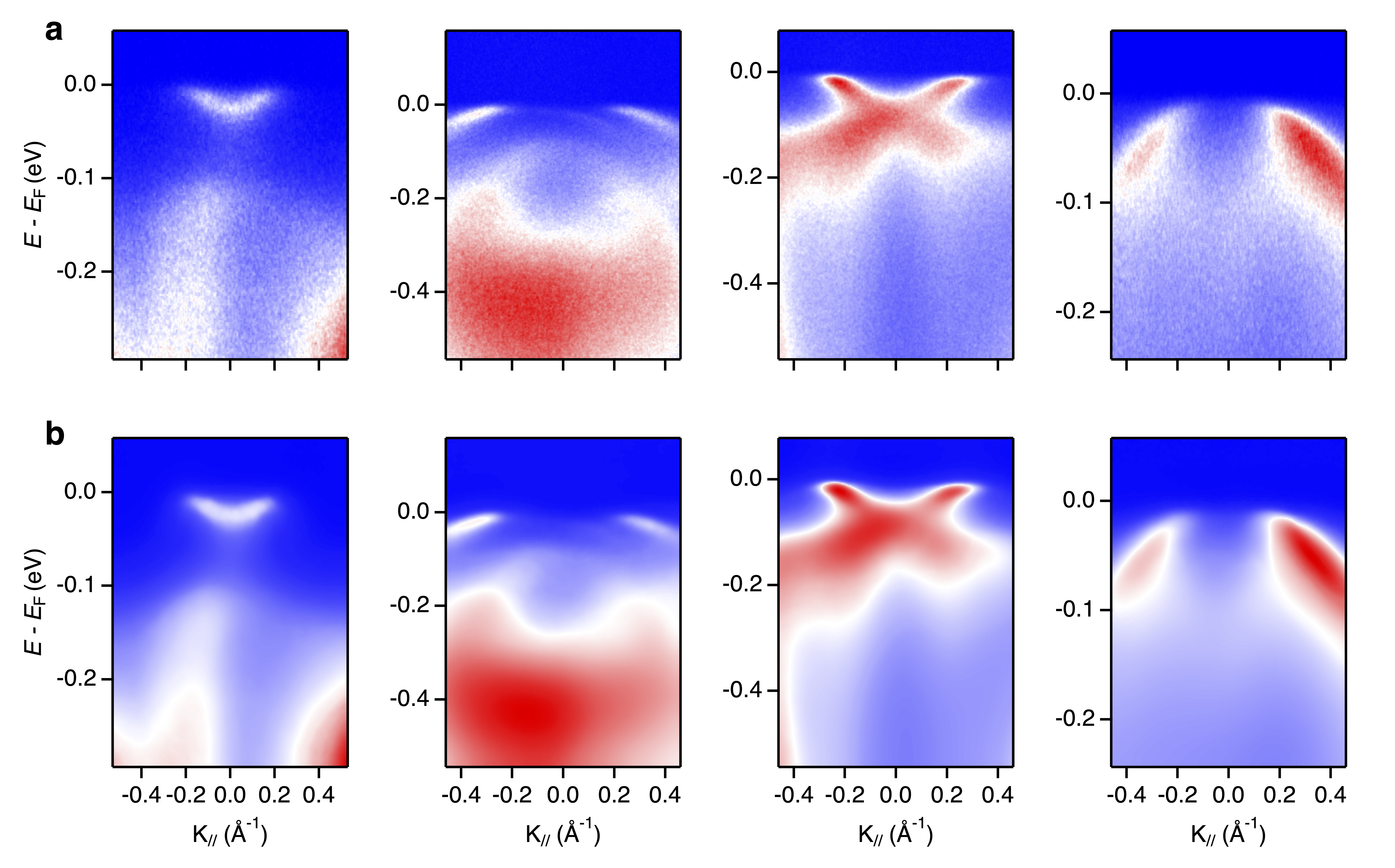}			
	\caption{More results of de-noised ARPES spectra. (a) The noise corrupted raw data. (b) The corresponding de-noised data based on the original data in panel (a). }
	\label{Fig: More results}
\end{figure*}

\subsection{II. Solution of principle component pursuit}\label{sec:S2}
In this section, we show the optimization techniques for principle component pursuit (PCP) proposed in the main manuscript. 

The PCP problem aims to find a ``low-rank + sparse" decomposition for a given matrix. Mathematically, given a matrix $I$, we want to find a decomposition satisfying $I=L+S$ where $L$ is a low-rank matrix and the matrix $S$ is sparse. The sparsity here means that the matrix $S$ only has few non-zero elements and the other elements are zero. Under some mild conditions \cite{Candes2011}, the matrix $L$ and $S$ can be obtain by solving the so-called PCP problem:
\begin{equation}
    \min_{L,S} \norm{L}_* + \lambda \norm{S}_1 \quad   \text{s.t. } I = L+S,
    \label{Eq:PCP}
\end{equation}
where $\lambda>0$ is a parameter chosen by hand, $\norm{L}_*$ is the nuclear norm defined as the sum of all the singular values of $L$ and $\norm{S}_1$ is the so-called $\ell_1$ norm defined as the summation of the absolute values of all the entries in $S$.

When solving the problem \eqref{Eq:PCP}, simultaneously iterative updating $L$ and $S$ via gradient descent can be time-consuming and not easy. Note that the function is the summation of two separable terms: one only depends on $L$ and the other rely on $S$. Thus, we had better solve this problem by alternating direction method of multipliers (ADMM). This algorithm transfers the original problem into two easy problems and alternatively updates $L$ and $S$.

The augmented Lagrange with respect to Eq. \eqref{Eq:PCP} can be written as:
\begin{equation}
    \mathcal{L}(L,S,\Lambda,\mu) = \norm{L}_* + \lambda \norm{S}_1 + \langle \Lambda, L+S -I \rangle + \frac{\mu}{2}  \norm{L+S-I}_2^2,
\end{equation}
where $\langle\cdot\rangle$ represents the inner product, $\Lambda$ is the Lagrange multiplier in the matrix form and $\mu >0$ is a parameter with respect to the equality constrain.

Let $L_k,S_k$ and $\Lambda_k$ be the the value of $L, S$ and $\Lambda$ in the $k$-th iteration respectively, then the update rule of $L$ can be written as:
\begin{equation}
    \begin{split}
        L_{k+1} &= \argmin_L \mathcal{L}(L,S_{k},\Lambda_k,\mu) = \norm{L}_* + \lambda \norm{S_k}_1 + \langle \Lambda_k, L+S_k-I\rangle + \frac{\mu}{2}\norm{L+S_k-I}_2^2 \\
        & = \argmin_L  \norm{L}_* + \frac{\mu}{2} (2\langle \frac{\Lambda_k}{\mu}, L+S_k-I\rangle   + \norm{L+S_k-I}_2^2)\\
        & = \argmin_L \norm{L}_* + \frac{\mu}{2} \norm{L+S_k-I + \frac{\Lambda_k}{\mu}}_F^2 + \varphi(S_k,\Lambda_k) \\
        & = \argmin_L \norm{L}_* + \frac{\mu}{2} \norm{L+S_k-I + \frac{\Lambda_k}{\mu}}_F^2 \\
        & = D_{1/\mu}\left(I-S_k-\frac{\Lambda_k}{\mu}\right)
    \end{split}
\end{equation}
where $\varphi(S_k,\Lambda_k)$ is the remaining terms which do not contain $L_k$ and can be ignored in the optimization problem and $D_{1/\mu}[\cdot]$ is the singular value threshold function and can be derived via proximal mapping \cite{Parikh2014}. For any given matrix $W$, and function $D_{1/\mu}[W]$ takes the form:
\begin{equation}
    D_{1/\mu}[W] = U T_{1/\mu}(\Sigma) V^T,
\end{equation}
where $W=U\Sigma V^T$ is the singular value decomposition (SVD) of $W$, and the function $T_{1/\mu}(\cdot)$ is the so-called soft-threshold function \cite{Parikh2014} which yields:
\begin{equation}
    T_{1/\mu}(x) = \left \{ 
    \begin{array}{lr}
        x-1/\mu & x>1/\mu \\
        0 & -1/\mu <x < 1/\mu \\ 
        x + 1/\mu & x< -1/\mu.
    \end{array}
    \right.
\end{equation}

The update rule of matrix $S$ can be derived as follows:
\begin{equation}
    \begin{split}
        S_{k+1} &= \argmin_S \norm{L}_* + \lambda \norm{S}_1 + \langle \Lambda, L_k+S-I\rangle + \frac{\mu}{2}\norm{L_k+S-I}_2^2 \\
        & = \argmin_S \lambda \norm{S}_1 + \frac{\mu}{2} (2\langle \frac{\Lambda_k}{\mu}, L_k+S_k-I\rangle   + \norm{L_k+S_k-I}_2^2) \\
        & = \argmin_S \lambda \norm{S}_1 + \frac{\mu}{2} \norm{S_k -(I-L_k-\frac{\Lambda_k}{\mu})}_F^2 +\varphi(L_k,\Lambda_k)  \\
        & = \argmin_S \lambda \norm{S}_1 + \frac{\mu}{2} \norm{S_k -(I-L_k-\frac{\Lambda_k}{\mu})}_F^2 \\
        & = T_{\mu^{-1}\lambda} \left(I-L_k-\frac{\Lambda_k}{\mu}\right),
    \end{split}
\end{equation}
where $\varphi(L_k,\Lambda_k)$ is the remaining terms not containing $S_k$ and $T_{1/\mu}[\cdot]$ denotes the soft-threshold function mentioned above with parameter $\mu^{-1} \lambda$.

The update rule for Lagrange multiplier can be obtained directly:
\begin{equation}
    \Lambda_{k+1} = \Lambda_k + \frac{\partial L}{\partial \Lambda}\Big|_{\Lambda = \Lambda_k}  = \Lambda_k + \mu (L_k+S_k-I).
\end{equation}

\subsection{III. Proof of strict saddle point property}\label{sec:S3}
In this section, we first introduce the definition of strict saddle point and then derivation the strict saddle point property of the linear model mentioned in the main text:
\begin{equation}
	\min_{U,g,h} \mathcal{L}'(U,g,h) \equiv \frac{1}{4} \norm{UU^T + g\circ g - h \circ h -I}_2^2,
	\label{Eq:Simplified model}
\end{equation}
where $U \in \mathbb{R}^{n\times r}$ and $g,h \in \mathbb{R}^n$

The definition of strict saddle point can be defined as follows:
\begin{Definition}[]
    Let $f:\mathbb{R}^n \rightarrow \mathbb{R}$ be a twice continuously differentiable function.
    \begin{enumerate}
        \item A point $x^*$ is called a critical point if $x^*$ satisfies $\nabla f(x^*)=0$.
        \item A point $x^*$ is a strict saddle point of function $f$ if there exists a direction $d$ satisfying $d^T \nabla^2 f(x^*) d < 0$. Local maximizer is a special case of strict saddle point.
    \end{enumerate}    
\end{Definition}

For self-completeness, we restate the properties about critical point here:
\begin{Lemma}
    Let $r=UU^T + g\circ g - h \circ h -I$ be the residual, a critical point of $\mathcal{L}'$ is either a global minimal $(r=0)$ or noiseless $(g=h=0)$.
    \label{Lem:Critical Points}
\end{Lemma}
\begin{proof}
    The loss function $\mathcal{L}'$ can be written as:
    \begin{equation}
        \mathcal{L}' = \frac{1}{4} \norm{\sum_{i=1}^r u_i u_i^T + g\circ g - h\circ h -I}_2^2,
    \end{equation}
    where $u_i$ denotes the $i$-th column of $U$.

    We introduce $M \equiv - (g\circ g + h\circ h -I)$ for simplicity, and then the derivation with respect to $u_i(i=1,\cdots,r$ can be calculated directly:
    \begin{equation}
        \nabla_{u_i} \mathcal{L}' = \frac{1}{2}(u_iu_i^T + \sum_{j\neq i}^r u_j u_j^T - M)u_i + \frac{1}{2} (u_iu_i^T + \sum_{j\neq i}^r u_j u_j^T - M^T)u_i.
    \end{equation}
    Similarly, the gradient with respect to noise vector $g$ and $h$ has the form:
    \begin{equation}
    \begin{split}
        \nabla_g \mathcal{L}’ &=  r \circ g \\
        \nabla_h \mathcal{L}’ &= - r \circ h ,
    \end{split}
\end{equation}
where the operator $\circ$ is element-wise Hadamard product.

For a not-global-minimal critical point, we have $r\neq 0$. Then $\nabla_g \mathcal{L}'=0$ and $\nabla_h \mathcal{L}'=0$ give $g=h=0$.

\end{proof}

\paragraph{Rank-$1$ case}
As a informative example, we first prove the strict saddle point property for a easy rank-$1$ case.

When the rank of $U$ is $1$, the simplest problem is the two-dimensional case where $u=(u_1,u_2),g=(g_{11},g_{12},g_{21},g_{22}$ and $h=(h_{11},h_{12},h_{21},h_{22})$. Thus, the loss function \eqref{Eq:Simplified model} can be written as:
\begin{equation}
    \mathcal{L}' = \frac{1}{4}\sum_{i,j=1}^2 (u_iu_j + g_{ij}^2 - h_{ij}^2 - I_{ij})^2,
\end{equation}
when we have the gradient:
\begin{equation}
    \begin{split}
        \nabla_u \mathcal{L} &= \left(\begin{pmatrix}
            u_1^2 & u_1u_2  \\ u_1u_2 & u_2^2 
        \end{pmatrix} + \begin{pmatrix}
            g_{11}^2  & \frac{1}{2}(g_{12}^2 + g_{21}^2)  \\ \frac{1}{2}(g_{21}^2 + g_{21}^2) & g_{22}^2
        \end{pmatrix} - \begin{pmatrix}
            h_{11}^2  & \frac{1}{2}(h_{12}^2 + h_{21}^2)  \\ \frac{1}{2}(h_{12}^2 + h_{21}^2) & h_{22}^2
        \end{pmatrix} -\frac{1}{2}(I+I^T)  \right) \times \begin{pmatrix}
            u_1 \\ u_2 
        \end{pmatrix} \\ 
        & = \begin{pmatrix}
            u_1^3 + u_1u_2^2 \\ u_1^2 u_2 + u_2^3  
        \end{pmatrix} + \begin{pmatrix}
             (g_{11}^2-h_{11}^2 ) u_1 + \frac{1}{2}(g_{12}^2 + g_{21}^2 - h_{12}^2-h_{21}^2) u_2 \\ \frac{1}{2}(g_{12}^2 + g_{21}^2 - h_{12}^2-h_{21}^2) u_1 + (g_{22}^2 - h_{22}^2 ) u_2 
        \end{pmatrix} - \frac{1}{2}(I+I^T)u \\ 
        \frac{\partial \mathcal{L}}{\partial g} &= \left(\begin{pmatrix}
            u_1^2 & u_1u_2  \\ u_1u_2 & u_2^2 
        \end{pmatrix}  - \begin{pmatrix}
            h_{11}^2  & h_{12}^2  \\ h_{21}^2 & h_{22}^2
        \end{pmatrix} -I  \right)\circ g  + \begin{pmatrix}
            g_{11}^3  & g_{12}^3  \\ g_{21}^3 & g_{22}^3
        \end{pmatrix} \\
        \frac{\partial \mathcal{L}}{\partial h} &= -\left( \begin{pmatrix}
            u_1^2 & u_1u_2  \\ u_1u_2 & u_2^2 
        \end{pmatrix} + \begin{pmatrix}
            g_{11}^2  & g_{12}^2  \\ g_{21}^2 & g_{22}^2
        \end{pmatrix} - I  \right) \circ h + \begin{pmatrix}
            h_{11}^3  & h_{12}^3  \\ h_{21}^3 & h_{22}^3
        \end{pmatrix},
    \end{split}
\end{equation}
By taking second derivative, the Hessian can be written as:
\begin{equation}
    \begin{split}
        \nabla^2_u \mathcal{L} &= \begin{pmatrix}
            3u_1 + u_2^2 & 2u_1u_2 \\ 2u_1u_2 & u_1^2 + 3u_2^2 
        \end{pmatrix} + \begin{pmatrix}
            g_{11}^2-h_{11}^2  & \frac{1}{2}(g_{12}^2 + g_{21}^2 - h_{12}^2-h_{21}^2)  \\ \frac{1}{2}(g_{12}^2 + g_{21}^2 - h_{12}^2-h_{21}^2)  & g_{22}^2 - h_{22}^2 
        \end{pmatrix} - \frac{1}{2}(I+I^T)  \\
        \nabla_g\nabla_u \mathcal{L} &= \begin{pmatrix}
            2g_{11}u_1 & g_{12}u_2 & g_{21}u_2 &0 \\ 
            0 & g_{12}u_1 & g_{21}u_1 & 2g_{22}u_2 
        \end{pmatrix} \\ 
        \nabla_h\nabla_u \mathcal{L} &= \begin{pmatrix}
            -2h_{11}u_1 & -h_{12}u_2 & -h_{21}u_2 &0 \\ 
            0 & -h_{12}u_2 & -2h_{21}u_1 & -2h_{22}u_2 
        \end{pmatrix} \\ 
        \nabla_g^2 \mathcal{L} &= \begin{pmatrix}
            u_1^2 + g_{11}^2 - h_{11}^2 - I_{11}  & 0 & 0 &0 \\ 
            0 & u_1u_2 + g_{12}^2-h_{12}^2  - I_{12}& 0 & 0 \\ 
            0 & 0 & u_1u_2 + g_{21}^2- h_{21}^2   -I_{21} &0 \\ 
            0 &0 &0 & u_2^2 + g_{22}^2 - h_{22}^2  - I_{22}
        \end{pmatrix}  + 2 \mathrm{diag}(g^2) \\
        \nabla_h^2 \mathcal{L} &= -\begin{pmatrix}
            u_1^2 + g_{11}^2 - h_{11}^2 - I_{11}  & 0 & 0 &0 \\ 
            0 & u_1u_2 + g_{12}^2-h_{12}^2  - I_{12}& 0 & 0 \\ 
            0 & 0 & u_1u_2 + g_{21}^2- h_{21}^2   -I_{21} &0 \\ 
            0 &0 &0 & u_2^2 + g_{22}^2 - h_{22}^2  - I_{22}
        \end{pmatrix} + 2\mathrm{diag}(h^2) \\
                \nabla_g\nabla_h \mathcal{L} &= - 2\begin{pmatrix}
            h_{11}g_{11} & 0 & 0 &0 \\ 0 & h_{12}g_{12} &0 &0 \\ 0&0& h_{21}g_{21} &0 \\  0& 0 & 0 & h_{22}g_{22} 
        \end{pmatrix}  = - 2\mathrm{diag}(h\circ g) \\ 
        \nabla_h\nabla_g \mathcal{L} &= - 2\begin{pmatrix}
            h_{11}g_{11} & 0 & 0 &0 \\ 0 & h_{12}g_{12} &0 &0 \\ 0&0& h_{21}g_{21} &0 \\  0& 0 & 0 & h_{22}g_{22} 
        \end{pmatrix}  = - 2\mathrm{diag}(h\circ g) \\
                \nabla_u\nabla_g \mathcal{L} &= \begin{pmatrix}
            2u_1 g_{11} & 0 \\ 
            u_2 g_{12} & u_1 g_{12} \\ 
            u_2 g_{21} & u_1 g_{21} \\ 
            0 & 2u_2 g_{22} 
        \end{pmatrix} \\
        \nabla_u\nabla_h \mathcal{L}&= \begin{pmatrix}
            -2u_1 h_{11} & 0 \\ 
            -u_2 h_{12} & -u_1 h_{12} \\ 
            -u_2 h_{21} & -u_1 h_{21} \\ 
            0 & -2u_2 h_{22} 
        \end{pmatrix} 
    \end{split}
\end{equation}

We then consider the Hessian matrix at critical points. By lemma \ref{Lem:Critical Points}, we have $g=h=0$ and $r\neq 0$. Thus, at this point, the Hessian matrix only has diagonal terms and all the remaining crossing terms are zero. Then, the Hessian matrix becomes:

\begin{equation}
    \begin{split}
        \nabla^2_{u}\mathcal{L} &=\begin{pmatrix}
            3u_1 + u_2^2 & 2u_1u_2 \\ 2u_1u_2 & u_1^2 + 3u_2^2 
        \end{pmatrix}- \frac{1}{2}(I+I^T) \\ 
        \nabla_g\nabla_u \mathcal{L} & = \begin{pmatrix}
            0 & 0 &0 &0 \\ 0 & 0& 0 & 0 \\ 
        \end{pmatrix} \\
        \nabla_h\nabla_u \mathcal{L} & = \begin{pmatrix}
            0 & 0 &0 &0 \\ 0 & 0& 0 & 0 \\ 
        \end{pmatrix} \\
        \nabla^2_g \mathcal{L} &= \mathrm{diag}(r) \\
        \nabla^2_h \mathcal{L} &= -\mathrm{diag}(r) \\
        \nabla_h\nabla_g \mathcal{L}&= 0_{4\times 4} \\
        \nabla^2_h \mathcal{L} &= 0_{4\times 4}\\
        \nabla_u\nabla_g \mathcal{L} & = \begin{pmatrix}
            0 & 0 \\ 0 &0 \\ 0 & 0 \\  0 & 0 \\ 
        \end{pmatrix} \\
        \nabla_u\nabla_h \mathcal{L} & = \begin{pmatrix}
            0 & 0 \\ 0 &0 \\ 0 & 0 \\ 0 & 0, \\ 
        \end{pmatrix} \\
    \end{split}
\end{equation}
where $r = uu^T + g\circ g - h \circ h - I$ stands for the residual.

We then construct the direction $d$ at critical points. By Lemma \ref{Lem:Critical Points}, there exists an index $i\in \{1,2\}$ such that $g^i = h^i = 0$ and $r^i \neq 0$. Hence, only the following two cases can appears:
\begin{itemize}
    \item $r^i>0$, we choose $d_u = [0,0], d_g=[0,0]$ and $d_h =[1,0](i=1)$ or $d_h=[0,1](i=2)$. At this direction, the corresponding quadratic form of Hessian is :
    \begin{equation}
        d^T \nabla^2 \mathcal{L}' d  = -r^i < 0.
    \end{equation}
    \item $r^i<0$, we choose $d_u = [0,0], d_h=[0,0]$ and $d_g =[1,0](i=1)$ or $d_g=[0,1](i=2)$. At this direction, the quadratic form of Hessian with respect to the critical points is :
    \begin{equation}
        d^T \nabla^2 \mathcal{L}' d  = r^i < 0.
    \end{equation}    
\end{itemize}

\paragraph{General rank-$r$ case}
We then deal with the general problem \eqref{Eq:Simplified model}. It can be written as the following equivalent form
\begin{equation}
            \min_{\{u_i\}_{i=1}^r ,g,h}\mathcal{L} =  \frac{1}{2}\norm{\sum_{i=1}^r u_i u_i^T  + g\circ g - h \circ h -I}_F^2,
\end{equation}
 where $u_i$ denotes the $i$-th column of the matrix $U$.

By Lemma \ref{Lem:Critical Points}, there exists an index $i$ such that $r^i \neq 0$ and $g^i=h^i= 0$ holds at critical points. And the corresponding gradient can be written as:
\begin{equation}
    \begin{split}
        \nabla_{u_i} \mathcal{L} &= (u_iu_i^T + \sum_{j\neq i}^r u_j u_j^T - M)u_i + (u_iu_i^T + \sum_{j\neq i}^r u_j u_j^T - M^T)u_i \\
        \nabla_g \mathcal{L} &= 2 \mathrm{diag}\left(\sum_{j=1}^r u_iu_i^T -M\right) \circ g \\
        \nabla_h \mathcal{L} &= -2 \mathrm{diag}\left(\sum_{j=1}^r u_iu_i^T -M\right) \circ h,
    \end{split}
\end{equation}
where $i=1,\cdots, r$ and we again introduce the symbol $M \equiv - (g\circ g + h\circ h -I)$ for brevity.

The Hessian can not be written in a compact matrix form, we then show it in the component form. The first part is the diagonal part about $u_i$:
\begin{equation}
    \begin{split}
        \nabla_{u_i}(\nabla_{u_i}\mathcal{L}) &= \nabla_{u_i} \left( (u_iu_i^T + \sum_{j\neq i}^r u_j u_j^T - M)u_i + (u_iu_i^T + \sum_{j\neq i}^r u_j u_j^T - M^T)u_i \right) \\
        & = \norm{u_i}^2 I +  2 u_i u_i ^T + \sum_{j\neq i}^r u_j u_j^T - M + \norm{u_i}^2 I +  2 u_i u_i ^T + \sum_{j\neq i}^r u_j u_j^T - M^T \\
        & = 2 \norm{u_i}^2 I + 4 u_i u_i^T + 2 \sum_{j\neq i}^r u_j u_j^T - (M + M^T),
    \end{split}
\end{equation}
where $i=1,\cdots r$.

The second is about the crossing terms:
\begin{equation}
    \begin{split}
        \nabla_{u_j}\nabla_{u_i} \mathcal{L} &= \nabla_{u_j}\left( (u_ju_j^T + \sum_{k\neq j}u_k u_k^T - M)u_i + (u_ju_j^T + \sum_{k\neq j}u_k u_k^T - M^T)u_i \right) \\ 
        & = 2u_j^T u_i I + 2u_j u_i^T. 
    \end{split}
\end{equation}
where $i,j=1,\cdots r$.

The crossing terms about $g,h$ and $u$ is not easy to represent in the matrix form, we thus show it in the element-wise form:
\begin{equation}
    \begin{split}
        &\nabla_g \nabla_{u_i} \mathcal{L} = \nabla_g (-M u_i - M^T u_i) \\
        \Rightarrow &\nabla_{g_{kl}}(\partial_{u_i}\mathcal{L})_j = \nabla_{g_{kl}} \sum_{m}g^2_{jm}(u_i)_m + g^2_{mj}(u_i)_m,  \\
    \end{split}
\end{equation}
where $k,l = 1,\cdots ,n$ and $i=1,\cdots ,r$. Finally, we have 
\begin{equation}
    \begin{split}
        \nabla_{g_{kl}}(\partial_{u_i}\mathcal{L})_j = 2\sum_m \delta_{jk}\delta_{ml} (u_i)_m g_{jm} + \delta_{mk}\delta_{jl} g_{mj} (u_i)_m,
    \end{split}
\end{equation}
where $\delta_{jk}$ is the Kronecker delta function which takes the value $1$ only for $j=k$ and zero otherwise, $k,l = 1,\cdots ,n$ and $i=1,\cdots ,r$.

Similarly, we have 
\begin{equation}
    \begin{split}
        &\nabla_h \nabla_{u_i} \mathcal{L} = \nabla_h (-M u_i - M^T u_i) \\
        \Rightarrow &\nabla_{h_{kl}}(\partial_{u_i}\mathcal{L})_j = \nabla_{h_{kl}} \sum_{m}-h^2_{jm}(u_i)_m  - h^2_{mj}(u_i)_m  \\
    \end{split}
\end{equation}
and 
\begin{equation}
    \begin{split}
        \nabla_{h_{kl}}(\partial_{u_i}\mathcal{L})_j = -2\sum_m \delta_{jk}\delta_{ml} (u_i)_m h_{jm} + \delta_{mk}\delta_{jl} h_{mj} (u_i)_m,
    \end{split}
\end{equation}
where $k,l = 1,\cdots ,n$ and $i=1,\cdots ,r$.

The diagonal terms with respect to $g$ and $h$ follows similarly:
\begin{equation}
    \begin{split}
        \nabla_g^2 \mathcal{L} &= \nabla_g \left(2 \mathrm{diag}\left(\sum_{i=1}^r u_iu_i^T - M\right) \circ g \right) \\
        & = 2\nabla_g (\mathrm{diag} (g\circ g) \circ g) + 2 \mathrm{diag}\left(\sum_{i=1}^r u_iu_i^T - M \right) \\ 
        & = 2 \mathrm{diag}\left(\sum_{i=1}^r u_iu_i^T - M + 2 g^2 \right),
    \end{split}
\end{equation}
and 
\begin{equation}
    \begin{split}
        \nabla_h^2 \mathcal{L} &= \nabla_h \left(- 2 \mathrm{diag}\left(\sum_{i=1}^r u_iu_i^T - M\right) \circ h\right) \\
        & = \nabla_h (-2 (\mathrm{diag} -(h\circ h)) \circ h) - 2 \mathrm{diag}\left(\sum_{i=1}^r u_iu_i^T - M \right) \\ 
        & = 2 \mathrm{diag}\left(-\left(\sum_{i=1}^r u_iu_i^T - M\right) + 2 h^2\right).
    \end{split}
\end{equation}

The corssing terms can be calculated directly:
\begin{equation}
    \begin{split}
        \nabla_g \nabla_h \mathcal{L} &= \nabla_g \left(-2 \left(\sum_{i=1}^r u_i u_i^T -M \right)\circ h \right)  \\ 
        & = \nabla_g (-2 (g\circ g) \circ h) \\ 
        & = -4 \mathrm{diag} (g\circ h ),
    \end{split}
\end{equation}
and 
\begin{equation}
    \begin{split}
        \nabla_h \nabla_g \mathcal{L} &= \nabla_h \left(2 \left(\sum_{i=1}^r u_i u_i^T -M \right)\circ g \right)  \\ 
        & = \nabla_g (-2 (h\circ h) \circ g) \\ 
        & = -4 \mathrm{diag} (h\circ g).
    \end{split}
\end{equation}

We then construct the direction satisfying $d^T \nabla^2 \mathcal{L}' d <0$, and the construction is similar to the rank-$2$ case.

By Lemma \ref{Lem:Critical Points}, there exists an index $n\in \{1,2,\cdots n\}$ such that $g^n = h^n = 0$ and $r^n \neq 0$. Hence, only the following two cases can appears:
\begin{itemize}
    \item $r^n>0$, we choose $d_{u_i} = [0,0,\cdots,0] (i=1,\cdots, r), d_g=[0,0,\cdots,0]$ and $d_h =[0,\dots,d^n_h = 1,\dots,0]$. At this direction, the corresponding quadratic form of Hessian is :
    \begin{equation}
        d^T \nabla^2 \mathcal{L} d = - 2r^n <0.
    \end{equation}
    \item $r^n<0$, we choose $d_{u_i} = [0,0,\cdots,0] (i=1,\cdots, r), d_g =[0,\dots,d^n_g = 1,\dots,0]$ and $d_h =[0,\cdots,0]$. At this direction, the corresponding quadratic form of Hessian is :
    \begin{equation}
        d^T \nabla^2 \mathcal{L} d = 2r^n <0.
    \end{equation}    
\end{itemize}

We now complete the proof.

\subsection{IV. Implementation details}\label{sec:S4}
We use pytorch \cite{PyTorch} to implement our algorithm. The correlation part (clean spectra) is parameterized via U-Net \cite{Ronneberger2015} equipped with batch normalization (BN) \cite{Ioffe2015} every layer. In addition, the activation functions is chosen to be leaky-ReLU (lReLU) whose parameter is set to $0.2$. As mentioned in the main text, AdamW optimizer \cite{Loshchilov2018} is used for respecting the curvature of landscape and stable training and its iteration scheme with respect to $\theta$ is listed in Eq. \eqref{Eq:AdamW}.
\begin{equation}
    \begin{split}
        h_{\theta, t} &= \nabla_\theta \mathcal{L}(\theta_{t-1}) \\
        v_{\theta, t} &= \beta_1 v_{\theta,t-1} + (1-\beta_1) h_{\theta, t} \\
        s_{\theta, t} &= \beta_2 s_{\theta, t-1} + (1-\beta_2) h_{\theta, t}^2 \\ 
        \hat{v}_{\theta, t} &= v_{\theta, t}/(1-\beta_1^t) \\
        \hat{s}_{\theta, t} &= s_{\theta, t}/(1-\beta_2^t) \\
        \theta_{t+1} &= \theta_t - \eta_a \left( \frac{\hat{v}_{\theta t}}{\sqrt{\hat{s}_{\theta t}} + \epsilon} + \lambda \theta_t \right),
    \end{split}
    \label{Eq:AdamW}
\end{equation}
where we choose $\beta_1 = 0.9,\beta_2 = 0.999, \epsilon = 1e-8$ and $\lambda = 0.02$ for training neural network as the parameters suggested in \cite{Loshchilov2018}. The small noise networks noted by $g$ and $h$ are optimized with the same optimizer and parameter, and the corresponding can be obtained by replacing $\theta$ to $g,h$ respectively. In practice, we find that this network are stable to train without applying weight decay.
\begin{table}[ht]
\caption{Neural network architecture}
\label{Tab:arch}
\setlength\tabcolsep{7pt}
\begin{tabular}{ccc}
  \hline\noalign{\smallskip}
  Encoder network\\
  \noalign{\smallskip}\hline\noalign{\smallskip}
  Input spectra $x$ \\
  Conv2d, BN, $3 \times 3 \times 128$, stride=2, lReLU\\
  Conv2d, BN, $3 \times 3 \times 128$, stride=2, lReLU\\
  Conv2d, BN, $3 \times 3 \times 128$, stride=2, lReLU \\
  Conv2d, BN, $3 \times 3 \times 128$, stride=2, lReLU \\
  \noalign{\smallskip}\hline
  Decoder network\\
  \noalign{\smallskip}\hline
  Conv2d, BN, $3 \times 3 \times 128$, stride=1, lReLU \\
  Upsampling, Conv2d, BN, $3 \times 3 \times 128$, stride=1, lReLU \\
  Upsampling, Conv2d, BN, $3 \times 3 \times 64$, stride=1, lReLU \\   
  Upsampling, Conv2d, BN, $3 \times 3 \times 32$, stride=1, lReLU \\
  Conv2d, $1 \times 1 \times 1$, stride=1, sigmoid\\
  \noalign{\smallskip}\hline
\end{tabular}
\end{table}